\documentclass[12pt,print,draftcls,onecolumn,romanappendices]{ieeecolor}
\usepackage{generic}
\usepackage[noadjust]{cite}
\usepackage{setspace}
\usepackage{graphicx}
\usepackage{nccmath}
\usepackage{amsmath,amsfonts,amssymb}
\usepackage{mathtools}
\usepackage{url}
\usepackage{color}
\usepackage{enumerate}
\usepackage{subcaption}
\usepackage[format=hang]{caption}
\usepackage[hidelinks]{hyperref}
\usepackage{textcomp}
\def\BibTeX{{\rm B\kern-.05em{\sc i\kern-.025em b}\kern-.08em
		T\kern-.1667em\lower.7ex\hbox{E}\kern-.125emX}}
\usepackage{stfloats}
\usepackage[cal=boondoxo]{mathalfa}
\DeclareMathAlphabet{\pazocal}{OMS}{zplm}{m}{n}

\usepackage[ruled,linesnumbered]{algorithm2e}

\makeatletter
\newcommand{\removelatexerror}{\let\@latex@error\@gobble}
\makeatother

\newcounter{storealgline}
\IncMargin{.5em}

\DontPrintSemicolon
\makeatletter
\renewcommand{\Indentp}[1]{%
	\advance\leftskip by #1
	\advance\skiptext by -#1
	\advance\skiprule by #1}%
\renewcommand{\Indp}{\algocf@adjustskipindent\Indentp{\algoskipindent}}
\renewcommand{\Indm}{\algocf@adjustskipindent\Indentp{-\algoskipindent}}
\makeatother

\newtheorem{defn}{Definition}
\newtheorem{thm}{Theorem}

\newtheorem{prop}[thm]{Proposition}

\newtheorem{lem}[thm]{Lemma}
\newtheorem{rem}{Remark}

\DeclareMathOperator{\col}{col}
\DeclareMathOperator{\cs}{colspan}
\DeclareMathOperator{\rank}{rank}
\DeclareMathOperator{\diag}{diag}

\newcommand{\B}{\mathcal{B}}
\newcommand{\g}{\mathfrak{g}}
\newcommand{\Hk}{\mathcal{H}}

\newcommand{\F}{\mathcal{F}}

\newcommand{\n}[1]{\mathtt{n}\left(#1\right)}    
\newcommand{\lag}[1]{\mathtt{L}\left(#1\right)}

\usepackage{xifthen}
\newcommand{\sR}[2][]{\ifthenelse{\isempty{#1}}{\mathbb{R}^{#2}}{\mathbb{R}^{#1\times #2}}}

\newcommand{\bint}[2]{{|[#1,#2]}}

\newcommand{\revise}[1]{{\color{black} #1}}
\begin{document}

\title{The Intrinsic State Variable in Fundamental Lemma and Its Use in Stability Design for Data-based Control}

\author{Yitao Yan, Jie Bao and Biao~Huang
\thanks{Y. Yan and J. Bao are with the School of Chemical Engineering, UNSW Sydney, NSW 2052, Australia. (e-mail:  y.yan@unsw.edu.au; j.bao@unsw.edu.au).}
\thanks{B. Huang is with the Department of Chemical and Materials Engineering, University of Alberta, 116 St. and 85 Ave., Edmonton, AB, Canada T6G 2R3. (e-mail: biao.huang@ualberta.ca)}
}
\markboth{Draft}{}

\maketitle
\global\csname @topnum\endcsname 0
\global\csname @botnum\endcsname 0
\begin{abstract}
	In the data-based setting, analysis and control design of dynamical systems using measured data are typically based on overlapping trajectory segments of the input and output variables. This could lead to complex designs because the system internal dynamics, which is typically reflected by the system state variable, is unavailable. In this paper, we will show that the coefficient vector in a modified version of Willems' fundamental lemma is an intrinsic and observable state variable for the system behavior. This argument evolves from the behavioral framework without the requirement of prior knowledge on the causality among system variables or any predefined representation structure (e.g., a state space representation). Such a view allows for the construction of a state map based on the fundamental lemma, bridging the trajectory space and the state space. The state property of the coefficient vector allows for a simple stability design approach using memoryless quadratic functions of it as Lyapunov functions, from which the control action for each step can be explicitly constructed. Using the coefficient vector as a state variable could see wide applications in the analysis and control design of dynamical systems including directions beyond the discussions in this paper.
\end{abstract}


\IEEEpeerreviewmaketitle
\allowdisplaybreaks

\section{Introduction}\label{sec:introduction}
The fundamental lemma developed in \cite{Willems:2005} has attracted significant attention in the past decade for data-based analysis and control. It states that the behavior of a controllable linear time-invariant (LTI) dynamical system can be parameterized by the columns of a Hankel matrix formed by one of its trajectories, provided that certain persistent excitation on the trajectory holds. This lemma has been further developed and improved in various aspects including the relaxation of the controllability and persistent excitation requirements \cite{Markovsky:2023}, the introduction of other equivalent representations \cite{Markovsky:2023,Yan:2023} and the representation using multiple trajectories in case of insufficient excitation on individual ones \cite{vanWaarde:2020}. In the case when noise is present in the measured trajectories, various representations have been developed to approximate the true system behavior including low rank approximation \cite{Markovsky:2012}, posing rank constraints \cite{Markovsky:2022}, trajectory averaging \cite{Sassella:2021} and the use of noise statistical properties \cite{Yan:2023a}. There have also been developments beyond stand-alone LTI systems including interconnected LTI system networks \cite{Yan:2023,Allibhoy:2021}, linear parameter-varying systems \cite{Toth:2011} and nonlinear systems \cite{Berberich:2022,Alsalti:2021,Han:2024}. On another direction, a stochastic version of the lemma has been developed in \cite{Pan:2022} using polynomial chaos expansion.

In addition to the developments on the lemma itself, this data-based representation has become the foundation on which many recent data-based control approaches are developed. A review of recent developments can be found in \cite{Markovsky:2021}. One of the most notable approaches is the data-enabled predictive control (DeePC) algorithm \cite{Coulson:2019}, which is based on the data-driven control algorithm in \cite{Markovsky:2008} with certain regulatory constraints to handle the noise in the measurements. This algorithm has later been robustified in \cite{Berberich:2020a} to guarantee stability in the case of small measurement noises. In the case with larger noise, based on a sufficiently accurate approximate behavior (e.g., constructed via the methods in \cite{Yan:2023a}), an optimal control algorithm has recently been developed in \cite{Yan:2024}. Outside control design, there have also been developments in other directions including dissipativity analysis \cite{Romer:2019}, trajectory filtering \cite{Alpago:2020} and data-based fault diagnosis \cite{Li:2022}, as well as applications to physical systems \cite{Huang:2021,Elokda:2021,Ozmeteler:2024}.

Being a result rooted from the behavioral systems theory, the fundamental lemma inherits various characteristics of the behavioral framework such as the focus on system external variables (e.g., manipulated variables and controlled variables) without any structural constraints on the representation, the need for state information or causality assumptions among the variables. This leads to a general parameterization of the system trajectory space constructed directly from data, whose direction of signal flow among variables may be unclear and state measurement is unavailable. As the state variable encodes the system memory, the lack of state information needs to be compensated by the use of sufficiently long trajectory. This has lead to the development of various theories (e.g., Lyapunov stability theory ) focusing on only the external variables \cite{Willems:1998}. While they provide valuable insights, control design based on these theories is a more complicated task because the Lyapunov functions need to be defined on the system trajectories, and the length of the trajectory segments required to construct these functions becomes a trial-and-error problem in practice. This approach also leads to a less structured Lyapunov function (e.g., positive semi-definiteness instead of positive definiteness \cite{Willems:1998}), making control design more difficult. Other than approaches that assume the availability of the state measurements (e.g., \cite{daSilva:2018}), there are approaches that carry out control design by using system trajectories as fictitious states \cite{dePersis:2019,Alpago:2020} for LTI systems with certain assumptions (e.g., strictly proper systems).

In this paper, we show that the fundamental lemma is naturally equipped with an \emph{intrinsic} state variable. This argument is carried out in the behavioral framework, and does not pose any structural constraints on the system representation or causality among system variables (e.g., requirement on the \emph{a priori} partition of input and output variables). Such a view leads to a simple and well-structured stability design approach using Lyapunov functions in the form of memoryless quadratic functions of this state variable, from which the controller behavior can be directly obtained.

The rest of this paper is organized as follows. Section \ref{sec:preliminaries} introduces relevant concepts in the behavioral systems theory. Section \ref{sec:stateProof} proves the state property of the parameterizer and develops its state map. These are the key ingredients for the stability design based on memoryless Lyapunov functions, which is presented in Section \ref{sec:stability}. An illustrative example is provided in Section \ref{sec:example} to show how the proposed approach is used. Finally, we conclude this paper and point to possible future directions to which this approach can be applied.

{\bf Notations.} We use the standard notations $\mathbb{R}$, $\mathbb{R}^\mathsf{n}$, $\mathbb{R}^{\mathsf{m}\times\mathsf{n}}$, etc. For a vector $w$, its dimension is denoted by $\mathsf{w}$. The set of non-negative integers is denoted as $\mathbb{Z}_{\geq0}$. An $\mathsf{n\times n}$ identity matrix and an ${\mathsf{m}\times\mathsf{n}}$ zero matrix are denoted as $I_\mathsf{n}$ and $0_{\mathsf{m}\times\mathsf{n}}$, respectively, and their subscripts are omitted when they are clear from context. For a matrix $A$, its transpose, inverse (if it exists) and Moore-Penrose inverse are denoted as $A^\top$, $A^{-1}$ and $A^\dagger$, respectively. Denote $A^\perp\coloneqq I-A^\dagger A$. For two matrices $A$ and $B$ with compatible dimensions, we denote $\col(A,B)\coloneqq\begin{bmatrix}
    A^\top & B^\top
\end{bmatrix}^\top$. The block diagonal matrix formed by matrices $A$ and $B$ is denoted by $\diag(A,B)$. Finally, the quadratic form $w^\top Mw$ with $M=M^\top$ is denoted as $\|w\|_M^2$.
\section{Preliminaries}\label{sec:preliminaries}
A dynamical system in the behavioral framework is defined as a triple $\Sigma=(\mathbb{T},\mathbb{W},\B)$, where $\mathbb{T}$ is the time axis, $\mathbb{W}$ is the signal space, and $\B\subset\mathbb{W}^\mathbb{T}$ is the system behavior \cite{Willems:1991}. In this paper, unless otherwise specified, we assume that $\mathbb{T}\subset\mathbb{Z}_{\geq0}$. A trajectory of the system generic variable $w$ (which is referred to as the \emph{manifest variable}) is denoted as $\tilde{w}$. The segment of $\tilde{w}$ in the interval $[k_1,k_2]$ is denoted as $\tilde{w}_\bint{k_1}{k_2}$, and the set of segments for all $\tilde{w}\in\B$ forms the behavior restricted to the interval $[k_1,k_2]$, denoted as \cite{Willems:2005}
\begin{equation}
    \B_\bint{k_1}{k_2}=\left\{\tilde{w}\mid\exists\tilde{w}'\in\B, \tilde{w}=\tilde{w}'_\bint{k_1}{k_2}\right\}.
\end{equation}
A dynamical system is time-invariant if $\sigma\B\subset\B$, where $\sigma$ is a shifting operator such that $\sigma w_k=w_{k+1}$. Time-invariant systems have finite memory spans, i.e., for a given time-invariant system, it takes a finite number of time steps before the past and future trajectories become independent. The smallest number of steps required to achieve this independence is called the \emph{lag} of the system and denoted as $\lag{\B}$. This characteristic of time-invariant systems allows for their trajectory segments to be weaved, forming longer ones, according to the following lemma \cite{Markovsky:2005}.
\begin{lem}[Trajectory Weaving]\label{lem:weaving}
Let $\Sigma$ be a time-invariant system and let $\tilde{w}^1,\tilde{w}^2\in\B$. If $\tilde{w}_{\bint{k-l}{k-1}}^1=\tilde{w}_{\bint{k-l}{k-1}}^2$ with $l\geq\lag{\B}$, then
\begin{equation}
    \tilde{w}_{\bint{0}{k-1}}^1\wedge\tilde{w}_{\bint{k}{\infty}}^2\in\B,
\end{equation}
where $\wedge$ denotes the concatenation of the two trajectory segments.
\end{lem}

Let $\Sigma=(\mathbb{T},\mathbb{W}_1\times\mathbb{W}_2,\B)$ be a dynamical system with manifest variable partitioned as $w=(w_1,w_2)$. If, for any $\tilde{w}_2\in\mathbb{W}_1^\mathbb{T}$, there exists $\tilde{w}_2$ such that $(\tilde{w}_1,\tilde{w}_2)\in\B$, then $w_1$ is a \emph{free} component of $w$ \cite{Willems:1991,Polderman:1998}. Such a partition can be an input/output partition if all elements in $w_1$ are free while none in $w_2$ is. An important concept associated with a partitioned manifest variable is the observability of one component from the other. In the behavioral framework, $w_2$ is observable from $w_1$ if, for two trajectories $(\tilde{w}_1^1,\tilde{w}_2^1), (\tilde{w}_1^2,\tilde{w}_2^2)\in\B$, $\tilde{w}_1^1=\tilde{w}_1^2$ implies $\tilde{w}_2^1=\tilde{w}_2^2$. In other words, trajectories of $w_2$ can be uniquely identified from those of $w_1$.

The behavior of a dynamical system can also be defined with the aid of an auxiliary variable called the \emph{latent variable}. In such a case, the full system is the quadruple $\Sigma^{full}=(\mathbb{T},\mathbb{W},\mathbb{L},\B^{full})$, where $\B^{full}\subset(\mathbb{W}\times\mathbb{L})^\mathbb{T}$ is the full behavior, whose corresponding manifest behavior is $\B=\{\tilde{w}\mid\exists\tilde{\ell}, (\tilde{w},\tilde{\ell})\in\B^{full}\}$, where $\ell$ is the latent variable. The state variable is an important class of latent variable that possess the property of state, which enables the  determination of whether two trajectories of a dynamical system can be concatenated together. The original definition of the state property by Willems (e.g., in \cite{Willems:1991} and \cite{Polderman:1998}) is given in continuous time. We therefore define a discrete-time counterpart of the state property that can be used in the data-based setting.
\begin{defn}[Property of State]\label{defn:state}
Denote $\Sigma^{full}=(\mathbb{T},\mathbb{W},\mathbb{L},\B^{full})$ as a latent variable dynamical system. The latent variable $\ell$ has the property of state if, for two trajectories $(\tilde{w}^1,\tilde{\ell}^1), (\tilde{w}^2,\tilde{\ell}^2)\in\B^{full}$, that $\ell_{k-1}^1=\ell_{k-1}^2$ implies 
\begin{equation}
    (\tilde{w}^1_\bint{0}{k-1}\wedge\tilde{w}^2_\bint{k}{\infty},\tilde{\ell}^1_\bint{0}{k-1}\wedge\tilde{\ell}^2_\bint{k}{\infty})\in\B^{full}
\end{equation}
for all $k\in\mathbb{T}$.
\end{defn}
For a time-invariant system, the smallest dimension of the state variable is called the state cardinality and denoted as $\n{\B}$.

A dynamical system is linear time-invariant (LTI) if, in addition to time-invariance, the signal space $\mathbb{W}$ is a vector space and $\B$ is a linear subspace of $\mathbb{W}^\mathbb{T}$ that is closed in the topology of pointwise convergence \cite{Willems:2005}. For such an LTI system $\Sigma$, its $(L+1)$-step behavior can be represented as the column span of an $(L+1)$-step Hankel matrix constructed from one of its $(T+1)$-step trajectories, denoted as
\begin{equation}\label{eq:HankelMatrix}
    \Hk=\begin{bmatrix}
		w_0 & w_1 & \cdots & w_{T-L}\\
        w_1 & w_2 & \cdots & w_{T-L+1}\\
		\vdots  & \vdots & \ddots & \vdots\\
		w_L & w_{L+1} & \cdots& w_T
	\end{bmatrix}\in\mathbb{R}^{(L+1)\mathsf{w}\times\mathsf{g}},
\end{equation}
where $\mathsf{g}=T-L+1$. This has recently become a foundation for many developments of data-based analysis and control design approach, and is known as the \emph{fundamental lemma} \cite{Willems:2005}. Since its initial appearance, the lemma has been further modified and improved in various ways (e.g., \cite{Markovsky:2023}). The key results useful to this paper are summarized as follows.
\begin{lem}[Fundamental Lemma]\label{lem:fundLemma}
Let $\Sigma$ be an LTI system whose manifest variable admits an input/output partition $w=(u,y)$. Let $\tilde{w}_\bint{0}{T}\in\B_\bint{0}{T}$. Assuming that $L\geq\lag{\B}$, then $\cs(\Hk)=\B_\bint{0}{L}$ if and only if $\rank(\Hk)=(L+1)\mathsf{u}+\n{\B}$. In such a case, for any trajectory $\tilde{v}\in\B_\bint{0}{L}$, there exists $g\in\sR{\mathsf{g}}$ such that
	\begin{equation}\label{eq:fundLemma}
		\tilde{v}=\Hk g.
	\end{equation}
\end{lem}
Lemma \ref{lem:fundLemma} provides a way to parameterize all $(L+1)$-step trajectories by choosing a the value of $g$, which is hereinafter referred to as the \emph{parameterizer}. The assumption of sufficiently long trajectory segment in Lemma \ref{lem:fundLemma} is crucial in trajectory simulation and predictive control algorithms (e.g.,  \cite{Markovsky:2008}). We therefore assume for the rest of this paper that $L\geq\lag{\B}$.

\section{parameterizer as a State Variable}\label{sec:stateProof}
Let $\Sigma=(\mathbb{T},\mathbb{W},\B)$ be an LTI system. Due to its time-invariant nature, \eqref{eq:fundLemma} holds true for any segments in a trajectory $w\in\B$ provided that its length is larger than $\lag{\B}$. Let $\tilde{w}_k$ denote the trajectory segment $\tilde{w}_\bint{k-L}{k}\in\B_\bint{k-L}{k}$ at the $k$th interval $[k-L,k]$. This segment can also be parameterized using the columns of $\Hk$ as
\begin{equation}\label{eq:wFullBehavior}
    \tilde{w}_k=\Hk g_k,
\end{equation}
where $\Hk$ is given by \eqref{eq:HankelMatrix}. By applying Lemma \ref{lem:weaving} recursively from $k=0$ (with $\tilde{w}_\bint{-L}{-1}$ as the initial trajectory) using $\tilde{w}_\bint{k-L}{k-1}$ as the overlapping steps between the $(k-1)$-th and the $k$th segment and aligning $w_k$ with $g_k$, a trajectory pair $(\tilde{w},\tilde{g})$ can be constructed with $\tilde{w}$ a trajectory of $\B$. This means that \eqref{eq:wFullBehavior} defines a full behavior for a latent variable dynamical system $\Sigma^{full}=(\mathbb{T},\mathbb{W},\mathbb{G},\B^{full})$ with $w$ the manifest variable, $g$ the latent variable, and $\B$ the corresponding manifest behavior. Furthermore, let the singular value decomposition (SVD) of $\Hk$ be
\begin{equation}
    \Hk=\begin{bmatrix}
        U_1 & U_2
    \end{bmatrix}\begin{bmatrix}
        S & 0\\ 0 & 0
    \end{bmatrix}
    \begin{bmatrix}
        V_1^\top \\ V_2^\top
    \end{bmatrix}
\end{equation}
with $S$ a diagonal matrix containing all non-zero singular values of $\Hk$. Define $\F=U_1$ (which has full column rank) and \revise{$\g=SV_1^\top g$}. Then the full behavior represented by
	\begin{equation}\label{eq:fullColRank}
		\tilde{w}_k=\F \g_k
	\end{equation}
and that represented by \eqref{eq:wFullBehavior} have the same manifest behavior, making the parameterizer $\g$ a latent variable of the system \cite{Yan:2023a}.

We now show that $\g$, in addition to being a latent variable, also has the property of state, and is hence a state variable for $\B$.
\begin{thm}\label{thm:gState}
    Let $\Sigma$ be an LTI system whose full behavior $\B^{full}$ is represented by \eqref{eq:fullColRank}. The following statements are true.
    \begin{enumerate}[(i)]
        \item $\g$ is a state variable for $\B$;
        \item $\g$ is observable from $w$.
    \end{enumerate}
\end{thm}
\begin{proof}
(i) Consider $\B^{full}$ in the interval $[k-L,k-1]$. Its segment $\tilde{w}_\bint{k-L}{k-1}$ can be represented as
\begin{equation}\label{eq:trajTransit}
    \tilde{w}_\bint{k-L}{k-1}=\Pi_p\tilde{w}_{k-1}=\Pi_p\F \g_{k-1},
\end{equation}
where $\Pi_p=\begin{bmatrix}
        0_{L\mathsf{w}\times\mathsf{w}} & I_{L\mathsf{w}}
    \end{bmatrix}$. Substituting into the representation of $\B^{full}_\bint{k-L}{k}$ and rearranging give 
\begin{equation}\label{eq:stateRep}
	\F \g_k-\begin{bmatrix}
		\Pi_p\F\\ 0 
	\end{bmatrix}\g_{k-1}-\begin{bmatrix}
		0\\I
	\end{bmatrix}w_k=0.
\end{equation}
Clearly, this representation shows that $\lag{\B^{full}}\leq 1$. Now, let $(\tilde{w}^1,\tilde{\g}^1)$, $(\tilde{w}^2,\tilde{\g}^2)\in\B^{full}$. If $\g_{k-1}^1=\g_{k-1}^2$, then $\tilde{w}_{k-1}^1=\tilde{w}_{k-1}^2$ by Lemma \ref{lem:fundLemma}, and, in particular, $w_{k-1}^1=w_{k-1}^2$. By Lemma \ref{lem:weaving}, the minimum number of overlapping steps required for trajectory weaving is at most 1, which means that $(\tilde{w}^1_\bint{0}{k-1}\wedge\tilde{w}^1_\bint{k}{\infty},\tilde{\g}^1_\bint{0}{k-1}\wedge\tilde{\g}^1_\bint{k}{\infty})\in\B^{full}$. As such, $\g$ has the property of state by Definition \ref{defn:state} and is therefore a state variable for $\B$.

(ii) Let $(\tilde{w}^1,\tilde{\g}^1)$, $(\tilde{w}^2,\tilde{\g}^2)\in\B^{full}$. If $\tilde{w}_1=\tilde{w}_2$, then $\tilde{w}^1_k=\tilde{w}^2_k$, which means that $\F(\g^1_k-\g^2_k)=0$. Since $\F$ is of full column rank, this only holds when $\g^1_k=\g^2_k$. Extending this illustration for all $k\in\mathbb{Z}_{\geq0}$ shows that $\tilde{w}^1=\tilde{w}^2$ implies $\tilde{\g}^1=\tilde{\g}^2$, hence $\g$ is observable from $w$.
\end{proof}
\begin{rem}
    The proof of the state property of $\g$ in \eqref{eq:fullColRank} also applies to the parameterizer $g$ in \eqref{eq:wFullBehavior}. However, due to the rank deficiency of $\Hk$, the mapping from $g_k$ to $\tilde{w}_k$ is only injective, which means that $g$ is not observable from $w$, making $g$ a less desirable state variable.  As such, the developments in the rest of this paper are centred around $\g$.
\end{rem}

Theorem \ref{thm:gState} shows that the parameterizer $\g$ is a state variable that naturally arises from the fundamental lemma, and is therefore an intrinsic state variable for the system behavior. A closer observation of \eqref{eq:stateRep} shows that it is in the form of the state space representation with backward shifting of $\g$ instead of forward. Furthermore, this state variable does not imply any causality among the elements in variable $w$. 

    

Using the rank requirement of $\Hk$ in Lemma \ref{lem:fundLemma}, we see that the dimension of $\g$ in \eqref{eq:fullColRank} is $(L+1)\mathsf{u}+\n{\B}$. In other words, $\g$ carries the information of the entire manifest variable trajectory. Furthermore, the observability of $\g$ means that one can directly obtain a unique state map of $\g_k$ for any $\tilde{w}_k\in\B_\bint{k-L}{k}$ as 
\begin{equation}\label{eq:stateMap}
    \g_k=\F^\dagger\tilde{w}_k.
\end{equation}
The combination of \eqref{eq:fullColRank} and \eqref{eq:stateMap} leads to a bijective mapping between $\tilde{w}_k$ and $\g_k$, allowing the discussions on one of them to be directly transferred to the other. As will be demonstrated in the next few sections, this mapping will simplify both the control design and implementation procedures significantly.

\section{Parameterizer-based Lyapunov Functions for Stability Design}\label{sec:stability}
\revise{In the behavioral framework, stability is defined on the manifest variables. The behavior of a dynamical system is asymptotically stable if, for all $\tilde{w}\in\B$, $\lim_{k\rightarrow\infty}\|w_k\|_2=0$ \cite{Polderman:1998}.} Traditionally, many stability design approaches for LTI systems rely on the construction of state-dependent Lyapunov functions. This is, however, not possible in the data-based setting due to the lack of the state information, which has led to the development of quadratic Lyapunov functions defined on the manifest variable trajectories \cite{Willems:1998,Kojima:2005}. However, the length of trajectory segments required to construct such Lyapunov functions is only known to be bounded by the system lag $\lag{\B}$, with its exact value unknown. In such a case, the coefficient matrix of the Lyapunov function is in general only positive semi-definite with unknown number of zero eigenvalues, making control synthesis more difficult.

Since the parameterizer has the state property, the immediate question is therefore whether the data-based stability analysis and stabilization design can be carried out by searching for a memoryless \revise{(i.e., lag zero)}  parameterizer-based Lyapunov function. We will show in this section that this is indeed the case, and the design procedure flows naturally using the full behavior \eqref{eq:fullColRank}.
\subsection{Parameterizer-based Lyapunov Functions}\label{subsec:stability}
We begin by focusing on the stability analysis of autonomous LTI systems. A dynamical system is autonomous if its trajectories are entirely characterized by their initial segments, i.e., for $\tilde{w}^1, \tilde{w}^2\in\B$, if $\tilde{w}_\bint{k-L}{k-1}^1=\tilde{w}_\bint{k-L}{k-1}^2$, then $\tilde{w}_\bint{k-L}{\infty}^1=\tilde{w}_\bint{k-L}{\infty}^2$.

\begin{prop}\label{prop:LF}
	Let $\Sigma$ be an autonomous LTI system whose full behavior $\B^{full}$ is represented by \eqref{eq:fullColRank}. Behavior $\B$ is asymptotically stable if and only if there exists $M=M^\top>0$ such that
    \begin{equation}\label{eq:stabilityLMI}
        \begin{bmatrix}
            M & *\\
            M\F_p^\dagger\Pi_p\F & M
        \end{bmatrix}>0,
    \end{equation}
    where $\F_p$ is the first $L\mathsf{w}$ rows of $\F$, $\Pi_p$ is defined in \eqref{eq:trajTransit}, and $*$ denotes the symmetric counterpart of its transposed block. In such a case, the quadratic function $V=\|\g\|_M^2$ is a Lyapunov function for $\B$.
\end{prop}
\begin{proof}
    Since $\g$ is observable from $w$, we have that $w$ is asymptotically stable if and only if $\g$ is. For autonomous systems, the $k$th step $w_k$, and hence the trajectory segment $\tilde{w}_k$, is completely determined by $\tilde{w}_\bint{k-L}{k-1}$. Furthermore, it is not difficult to see that, for a time-invariant system, $\n{\B}\leq\lag{\B}\mathsf{w}$. This means that $\F_p$, with $L\mathsf{w}$ rows and $\n{\B}$ columns, is of full column rank.
    
    Using a similar construction of $\tilde{w}_\bint{k-L}{k-1}$ to that in \eqref{eq:trajTransit} with \eqref{eq:fullColRank} as the representation, the $k$th step parameterizer $\g_k$ can be computed as
	\begin{equation}\label{eq:stabilityRecur}
		\F_p\g_k=\Pi_p\F \g_{k-1} \ \Rightarrow \ \g_k=\F_p^\dagger\Pi_p\F \g_{k-1}.
	\end{equation}
    It then follows that $\g$ is asymptotically stable if and only if all eigenvalues of $\F_p^\dagger\Pi_p\F$ are within the unit disk. This is well-known to be equivalent to the existence of a symmetric matrix $M>0$ such that
    \begin{equation}
        M-(M\F_p^\dagger\Pi_p\F)^\top M^{-1}(M\F_p^\dagger\Pi_p\F)>0,
    \end{equation}
    which is equivalent to \eqref{eq:stabilityLMI}. Pre- and post-multiplying $\g_{k}^\top$ and $\g_{k}$ lead to the inequality
    $\|\g_k\|_M^2-\|\g_{k-1}\|_M^2<0$ for all non-zero $\g_k$ and $\g_{k-1}$. Combining this with the positive-definiteness of $M$ shows that $V_k$ is a Lyapunov function of $\B$. 
\end{proof}

The proof of Proposition \ref{prop:LF} shows that, in the case of autonomous systems, the parameterizer $\g$ exhibits a similar structure to the conventional state representations. In fact, the dimension of $\g$ in this case is $\n{\B}$ (\revise{$\mathsf{u}$ in Lemma \ref{lem:fundLemma} is zero because autonomous systems do not have free components}), which means that it is a minimal state variable for $\B$.

\begin{rem}
    If the representation \eqref{eq:wFullBehavior} were used, then the condition in Proposition \ref{prop:LF} is only a (possibly rather conservative) sufficient condition because the asymptotic stability of $w$ is not necessarily equivalent to that of the parameterizer $g$ in \eqref{eq:wFullBehavior}. It is only through the use of $\g$ as a state can a necessary and sufficient condition for asymptotic stability be obtained.
\end{rem}

\subsection{Trajectory Stabilization Using parameterizer-based Control Lyapunov Function}\label{subsec:stabilization}
Control in the behavioral framework is viewed as an interconnection between the to-be-controlled system and the controller, and the controlled behavior is the common behavior of the two (and thus a \emph{sub-behavior} of the uncontrolled behavior). The key issue of control design is the feasibility of the controlled behavior, which consists of 3 components \cite{Willems:2002}, namely, \emph{existence}: whether the desired behavior is contained in the uncontrolled behavior; \emph{implementability}: whether the desired behavior can be implemented by the manipulated variable alone; and \emph{liveness}: whether the free components remain free in the controlled behavior. In the context of data-based receding horizon control, free components include past trajectory and future external variable (e.g., exogenous disturbance) trajectory, if there is any \cite{Yan:2023}.

We now consider the stabilization of LTI systems. In practice, it is desirable to stabilize all manifest variables (including manipulated variables) so that the system is internally stable. By a similar argument to that in Section \ref{subsec:stability}, the observable representation \eqref{eq:fullColRank} transforms the stabilization of $w$ into that of the parameterizer $\g$. As such, the stabilization problem is translated to the existence of a sub-behavior of $\g$ that is asymptotically stable.

Using the full behavior \eqref{eq:fullColRank}, the $k$th step parameterizer $\g_k$ can be obtained in a similar fashion to \eqref{eq:stabilityRecur} as
\begin{equation}\label{eq:gSelectStability}
    \g_k=\F_p^\dagger\Pi_p\tilde{w}_{k-1}+\F_p^\perp z_k'=\F_p^\dagger\Pi_p\F \g_{k-1}+\F_z z_k,
\end{equation}
where $\F_p$ and $\Pi_p$ are defined similarly as those in Proposition \ref{prop:LF}, $\F_z$ is a matrix of full column rank such that $\cs(\F_z)=\cs(\F_p^\perp)$ and $z_k$ can be understood as a virtual ``manipulated variable''. The existence of an asymptotically stable sub-behavior of $\g$ is therefore equivalent to the existence of $z_k$ that renders $\g$ in \eqref{eq:gSelectStability} asymptotically stable (i.e., the stabilisability of $\g$ in \eqref{eq:gSelectStability}).

\begin{prop}\label{prop:stabilisability}
    Let $\Sigma$ be an LTI system with $\B^{full}$ represented by \eqref{eq:fullColRank}. Let $\mathcal{S}\in\sR[\mathsf{g}]{\mathsf{g}}$ be a full rank matrix such that
    \begin{equation}
        \mathcal{S}^{-1}\F_p^\dagger\Pi_p\F \mathcal{S}=\begin{bmatrix}
            A_{11} & A_{12}\\ 0 & A_{22}
        \end{bmatrix}, \ \mathcal{S}^{-1}\F_z=\begin{bmatrix}
            B\\0
        \end{bmatrix},
    \end{equation}
    where $(A_{11},B)$ is controllable (i.e., the rank of the controllability matrix formed by $A_{11}$ and $B$ is of full row rank). Then there exists an implementable controlled behavior $\B_c\subset\B$ that is asymptotically stable if and only if all eigenvalues of $A_{22}$ are within the unit disk.
\end{prop}
\begin{proof}
    \revise{Let $\g'=\mathcal{S}^{-1}\g$. The dynamics of $\g'$ can be derived from \eqref{eq:gSelectStability} as
    \begin{equation}
        \begin{bmatrix}
            \g'_{1_k}\\ \g'_{2_k}
        \end{bmatrix}=\begin{bmatrix}
            A_{11} & A_{12}\\ 0 & A_{22}
        \end{bmatrix}\begin{bmatrix}
            \g'_{1_{k-1}}\\ \g'_{2_{k-1}}
        \end{bmatrix}+\begin{bmatrix}
            B\\0
        \end{bmatrix}z_k,
    \end{equation}
    where $\g_1'$ and $\g_2'$ are partitions of $\g'$ conforming to those of the coefficient matrices. Since $(A_{11},B)$ is controllable but $z_k$ has no effect on $\g_2'$, $\g'$ can be stabilized by $z$ if and only if $\g_2'$ is already asymptotically stable, which is equivalent to $A_{22}$ having all its eigenvalues within the unit disk. In such a case, any $z_k$ that stabilizes $\g_1'$ will stabilize $\g'$, and therefore $\g$.} By observability, the corresponding trajectory $\tilde{w}$ is also asymptotically stable, and the implementing manipulated variable in the $k$th interval is
    \begin{equation}\label{eq:ukSel}
        u_k=\Pi_u\F \g_k,
    \end{equation}
    where $\Pi_u$ selects $u_k$ from $\tilde{w}_k$.
\end{proof}

\begin{rem}
    While \eqref{eq:gSelectStability} resembles a state equation in the classical representation \revise{(i.e., of the form $x_{k+1}=Ax_k+Bu_k$), the underlying characteristics of the parameterizer is fundamentally different. Specifically, the virtual manipulated variable $z$ restricts the behavior of $\g$, which then selects the trajectories of all manifest variables simultaneously. In other words, there is no prior assumption of causality among manifest variables.}
\end{rem}

Proposition \ref{prop:stabilisability} validates the existence of a controlled behavior, but does not provide a  construction of the controller. As illustrated in Section \ref{sec:preliminaries}, the key issue of control design is to find a feasible controlled behavior. In this paper, we aim to find a controlled behavior that is directly parameterized by its free components because the existence of such a behavior automatically guarantees implementability and liveness. Since there is no external variable in this case, the only free component is the past trajectory $\tilde{w}_\bint{k-L}{k-1}$, which is parameterized by $\g_{k-1}$, it is sensible to construct a controller that renders the controlled behavior to have the representation
\begin{equation}\label{eq:controlledStateFB}
    \tilde{w}_k=\F_c \g_{k-1},
\end{equation}
where $\F_c$ is to be determined. This allows the $k$th segment of the controlled behavior, hence the $k$th step manipulated variable $u_k$, to be directly constructed from the state of the previous segment.
\begin{thm}\label{thm:stableStateFB}
    Let $\Sigma$ be an LTI system with $\B^{full}$ represented by \eqref{eq:fullColRank}. There exists a controller that renders the controlled behavior (in the form of \eqref{eq:controlledStateFB}) asymptotically stable if and only if there exist matrices $W=W^\top>0$ and $Y\in\sR[\mathsf{u}]{\mathsf{g}}$ such that
	\begin{equation}\label{eq:stabilityStateFBLMI}
			\begin{bmatrix}
			W & * \\
                \F_p^\dagger\Pi_p\F W+\F_zY &  W
			\end{bmatrix}>0.
	\end{equation}
In such a case, a solution for the manipulated variable $u_k$ based on previous step parameterizer $g_{k-1}$ is obtained as
\begin{equation}\label{eq:stabilityStateFBControl}
	u_k=\Pi_u\F(\F_p^\dagger\Pi_p\F+\F_zYW^{-1})\g_{k-1},
\end{equation}
where $\Pi_u$ is the same as that in \eqref{eq:ukSel}.
\end{thm}
\begin{proof}
Firstly, the controlled behavior should satisfy $\B_{c\bint{k-L}{k}}\subset\B_{\bint{k-L}{k}}$. This means that $\cs(\F_c)\subset\cs(\F)$, which can be achieved if and only if there exists a matrix $G\in\sR[\mathsf{g}]{\mathsf{g}}$ such that $\F_c=\F G$. As such, the existence of a controlled behavior in the form \eqref{eq:controlledStateFB} is equivalent to the existence of a matrix $G$ such that the transition $\g_k=G\g_{k-1}$ renders the controlled behavior asymptotically stable, which, by observability, is equivalent to the asymptotic stability of $\g$ itself. By Proposition \ref{prop:LF}, this is equivalent to the existence of a Lyapunov function  $V_k=\|\g_k\|_M^2$, where $M>0$, such that
\begin{equation}\label{eq:stabilityRate}
    \|\g_{k-1}\|_M^2-\|\g_k\|_M^2>0
\end{equation}
for all non-zero $\g_k$ and $\g_{k-1}$. Substituting $\g_k$ by \eqref{eq:gSelectStability} gives
\begin{equation}\label{eq:zExistStability}
     \|\g_{k-1}\|_M^2-\|\F_p^\dagger\Pi_p\F \g_{k-1}+\F_z z_k\|_M^2>0.
\end{equation}
The existence of an asymptotically stable controlled behavior is therefore transformed to the existence of $z_k$ such that \eqref{eq:zExistStability} holds for all $\g_{k-1}\neq0$. We now show that this is equivalent to the existence of a matrix $K$ such that $z_k=K\g_{k-1}$ satisfies \eqref{eq:zExistStability} for all $\g_{k-1}\neq0$.
	\begin{lem}\label{lem:gainConstruct}
		Let the behavior of $(v_1,v_2)$ be represented by
		\begin{equation}\label{eq:QFBehavior}
			\begin{bmatrix}
				v_1\\v_2
			\end{bmatrix}^\top\begin{bmatrix}
			Q & * \\ S^\top & -R
		\end{bmatrix}\begin{bmatrix}
		v_1\\v_2
	\end{bmatrix}\geq0
		\end{equation}
	 with $R>0$. Then $v_1$ is free if and only if there exists a matrix $K\in\mathbb{R}^{\mathsf{v}_2\times\mathsf{v}_1}$ such that $v_2=Kv_1$ satisfies \eqref{eq:QFBehavior} for all $v_1$.
	\end{lem}
\begin{proof}
    The \emph{if} part is obvious. For the \emph{only if} part, since $R>0$, \eqref{eq:QFBehavior} has an upper bound for any given value of $v_1$. Specifically, by choosing $v_2=R^{-1}S^\top v_1$, \eqref{eq:QFBehavior} reaches the maximum of $v_1^\top(Q+SR^{-1}S^\top)v_1$. As such, $v_1$ is free if and only if this maximum is non-negative for any $v_1$, which means that $Q+SR^{-1}S^\top\geq0$. In such a case, completing the square with respect to $v_2$ gives
	\begin{equation}
		(v_2-R^{-1}S^\top v_1)^\top R(v_2-R^{-1}S^\top v_1)\leq v_1^\top(Q+SR^{-1}S^\top)v_1.
	\end{equation}
	For a given $v_1$, all solutions of $v_2$ can then be computed by solving from
	\begin{equation}
		R^{1/2}(v_2-R^{-1}S^\top v_1)=U(Q+SR^{-1}S^\top)^{1/2}v_1,
	\end{equation}
	where $U\in\mathbb{R}^{\mathsf{v}_2\times\mathsf{v}_2}$ is any matrix such that $U^\top U\leq I$. As such, all solutions of $v_2$ that satisfy \eqref{eq:QFBehavior} for a given value of $v_1$ can be obtained as
	 \begin{equation}\label{eq:gainSol}
	 	v_2=[R^{-1/2}U(Q+SR^{-1}S^\top)^{1/2}+R^{-1}S^\top]v_1.
	 \end{equation}
 Note that the solution is in the form $v_2=Kv_1$ for some matrix $K$.
\end{proof}
Since $\F_z$ in \eqref{eq:gSelectStability} is of full column rank and $M>0$, we see that \eqref{eq:zExistStability} is in the form of (a stricter version of) \eqref{eq:QFBehavior} with $v_1\rightarrow \g_{k-1}$, $v_2\rightarrow z_k$ and $R\rightarrow \F_z^\top M\F_z$. The existence of $z_k$ in \eqref{eq:zExistStability} is therefore equivalent to the existence of $K$ such that $z_k=K\g_{k-1}$. Define the transformation $\g_k'=M\g_k$. We see that this transformation is invertible, i.e., $\g_k=M^{-1}\g_k'\eqqcolon W\g_k'$. The existence of $K$ such that $z_k=K\g_{k-1}$ satisfies \eqref{eq:zExistStability} for any non-zero $\g_{k-1}$ is therefore equivalent to the existence of $z_k=KW\g'_{k-1}\eqqcolon Y\g_{k-1}'$ such that
\begin{equation}
     \|\g_{k-1}'\|_W^2-\|(\F_p^\dagger\Pi_p\F W+\F_z Y)\g_{k-1}'\|_{W^{-1}}^2>0
\end{equation}
for all $\g_{k-1}'\neq0$. This is equivalent to
\begin{equation}
    W-(\F_p^\dagger\Pi_p\F W+\F_z Y)^\top W^{-1}(\F_p^\dagger\Pi_p\F W+\F_z Y)>0,
\end{equation}
which, in combination with $W=M^{-1}>0$, is equivalent to \eqref{eq:stabilityStateFBLMI}. In such a case, we have $z_k=Y\g_{k-1}'=YW^{-1}\g_{k-1}$. Substituting into \eqref{eq:gSelectStability} gives 
\begin{equation}\label{eq:stabilityStateFB}
    \g_k=(\F_p^\dagger\Pi_p\F+\F_zYW^{-1})\g_{k-1}.
\end{equation}
The solution \eqref{eq:stabilityStateFBControl} follows by substituting \eqref{eq:stabilityStateFB} to \eqref{eq:ukSel}. 
\end{proof}
\begin{rem}\label{rem:stableOpt}
    The solution \eqref{eq:stabilityStateFBControl} is only one of the possible solutions for a given $W$. Once the value of $W$ is obtained, one could use $M=W^{-1}$ to find all solutions for $z_k$ in \eqref{eq:zExistStability} using \eqref{eq:gainSol} so that the controlled behavior is of the form \eqref{eq:controlledStateFB}. In fact, if the problem is to minimize a certain cost function (e.g., an economic cost), then the results in Theorem \ref{thm:stableStateFB} can be used as constraints to ensure stability. Specifically, the existence of solutions to \eqref{eq:stabilityStateFBLMI} guarantees the feasibility of \eqref{eq:stabilityRate} for the system behavior. Using \eqref{eq:gSelectStability} as well as the state map \eqref{eq:stateMap} on the $(k-1)$-th step, \eqref{eq:stabilityRate} can be converted into 
\begin{equation}\label{eq:stableOptimization}
    \begin{bmatrix}
    \|\F^\dagger\tilde{w}_{k-1}\|_M^2 & *\\ \g_k^\top & M^{-1}
\end{bmatrix}>0,
\end{equation}
which is convex in $\g_k$.
\end{rem}

Theorem \ref{thm:stableStateFB} shows that one could synthesize a controlled behavior by finding a parameterizer-based control Lyapunov function, through which the controller can be directly obtained. The vector $\g_k$ therefore has dual functions of a parameterizer and a state variable. Furthermore, using the state map \eqref{eq:stateMap} on the $(k-1)$-th interval, \eqref{eq:stabilityStateFBControl} can be transformed into
 \begin{equation}\label{eq:stabilityTrajFBControl}
 \begin{split}
     u_k&=\Pi_u\F(\F_p^\dagger\Pi_p\F \g_{k-1}+\F_zYW^{-1}\g_{k-1})\\
     &=\Pi_u\F(\F_p^\dagger\Pi_p\tilde{w}_{k-1}+\F_zYW^{-1}\F^\dagger \tilde{w}_{k-1})\\
     &=\Pi_u\F(\F_p^\dagger\Pi_p+\F_zYW^{-1}\F^\dagger)\tilde{w}_{k-1},
     \end{split}
 \end{equation}
which means that, while the design is based on the parameterizer $\g$, the actual controller can be constructed based on the past trajectory. Interestingly, contrary to the classical control design approaches, in which the dynamic output feedback controller design is much more involved than state feedback, the controller based on the past trajectory in \eqref{eq:stabilityTrajFBControl} follows immediately from that built on the parameterizer. This is because classical input/state/output representation has inherent causality assumption while the behavior \eqref{eq:fullColRank} does not. Specifically, in the classical framework, the output trajectory is determined based on the initial state and the input trajectory, hence output feedback is conceptually a system (causality) inversion problem. Using the behavioral framework, the parameterizer alone identifies \emph{all} manifest variable trajectories in a symmetrical way, and thus the control action is a part of the manifest variable trajectory of the controlled behavior, not requiring solving an inversion problem.

\section{Illustrative Example}\label{sec:example}
We consider an LTI system $Y(z)=G(z)U(z)$, where 
\begin{subequations}
        \begin{align}
	G(z)=\begin{bmatrix}
		\frac{-0.2z+0.367}{z-1.083} & \frac{0.6775z+1.198}{z^2-0.324+0.449}\\
		\frac{-0.341z+0.449}{z^2-1.341z+0.449} & \frac{-0.428}{z-1.14}
	\end{bmatrix}
        \end{align}
\end{subequations}
Note that the system described by $G(z)$ is unstable and not minimum phase. Using Theorem \ref{thm:stableStateFB}, the manipulated variable trajectory is obtained from \eqref{eq:stabilityTrajFBControl}. The control performance is shown in Fig. \ref{fig:stability}, which demonstrates that all manifest variables have been asymptotically stabilized.
\begin{figure}[h]
	\centering
	\includegraphics[width=0.6\linewidth]{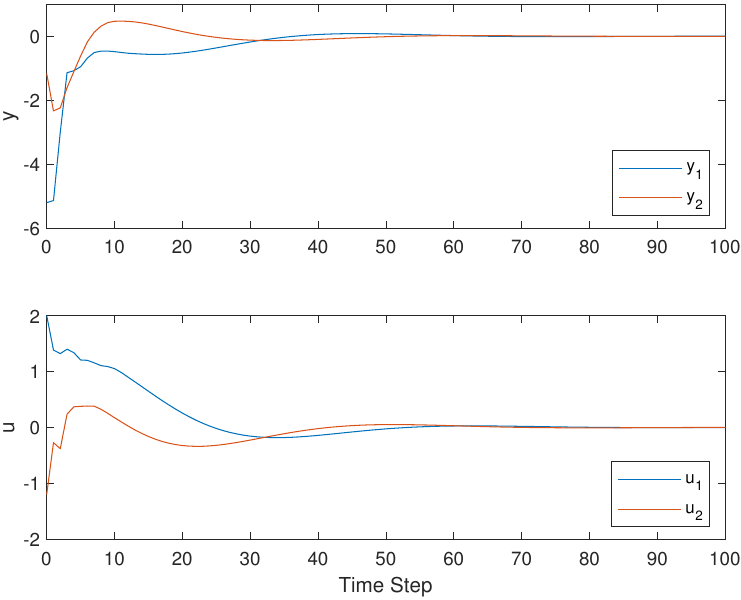}
	\caption{Stabilization Performance}
	\label{fig:stability}
\end{figure}

\section{Conclusion and Future Work}\label{sec:conclusion}
In this paper, we have shown that the parameterizer in Willem's fundamental lemma can be viewed as a state variable for the system behavior. The argument is established from the definition of the state property, and hence does not require system representation or \emph{a priori} input/output partition of the manifest variable. We have shown that the stability and dissipativity of the system can be analysed by constructing memoryless Lyapunov functions using the parameterizer, leading to an elegant approach for the synthesis of stable controlled behavior with the controller behavior explicitly constructed.

The use of the parameterizer as a state can potentially have many other applications beyond what are illustrated in this paper. One of the immediate directions is synthesis of dissipative behaviors for disturbance attenuation, in which case the storage functions, being an extension to Lyapunov functions, can potentially be represented by memoryless quadratic functions of the parameterizer as well. This could lead to a simpler design approach compared with that based on quadratic functions of the system trajectories (e.g., \cite{Yan:2023}). Another possible future direction is in the context of distributed control of interconnected systems, in which the simplification due to the use of state variables compared with that of trajectories may be significant. It is therefore interesting to investigate how the use of parameterizers can be integrated in distributed control design.

\bibliographystyle{IEEEtran}

\bibliography{refs}

\end{document}